\documentclass[a4paper, UKenglish]{lipics-v2016}

\def \draftmode {} 

\usepackage{microtype} 
\title{A rearrangement distance for fully-labelled trees}
\titlerunning{A rearrangement distance} 

\author[1]{Giulia Bernardini}
\author[1]{Paola Bonizzoni}
\author[1]{Gianluca Della Vedova}
\author[1]{Murray Patterson}
\affil[1]{DISCo, Universit\`{a} degli Studi Milano - Bicocca\\
	\texttt{giulia.bernardini@unimib.it},
	\texttt{paola.bonizzoni@unimib.it},
	\texttt{gianluca.dellavedova@unimib.it},
	\texttt{murray.patterson@unimib.it}}

\authorrunning{G. Bernardini et al.} 

\Copyright{Giulia Bernardini and Paola Bonizzoni and Gianluca Della Vedova and Murray Patterson} 


\EventEditors{John Q. Open and Joan R. Acces}
\EventNoEds{2}
\EventLongTitle{42nd Conference on Very Important Topics (CVIT 2016)}
\EventShortTitle{CVIT 2016}
\EventAcronym{CVIT}
\EventYear{2016}
\EventDate{December 24--27, 2016}
\EventLocation{Little Whinging, United Kingdom}
\EventLogo{}
\SeriesVolume{42}
\ArticleNo{23}

\usepackage{todonotes}
\usepackage{tikz}

\usetikzlibrary{
  arrows,
  automata,
}

\tikzset{
  every node/.style={
    align=center,
    draw=none,
  },
  -latex,
}

\newcommand{\ie}{\textit{i.e.}}
\newcommand{\eg}{\textit{e.g.}}

\newcommand{\cL}{\mathcal{L}}
\newcommand{\cT}{\mathcal{T}}
\newcommand{\cM}{\mathcal{M}}
\newcommand{\cC}{\mathcal{C}}
\newcommand{\cP}{\mathcal{P}}
\newcommand{\cX}{\mathcal{X}}
\newcommand{\cI}{\mathcal{I}}
\newcommand{\cS}{\mathcal{S}}

\newcommand{\tone}{T_1^H}
\newcommand{\ttwo}{T_2^H}

\newcommand{\edge}[3]{\ensuremath{#1 | #2 \rightarrow #3}}
\newcommand{\dm}{d_\Delta}

\begin{document}

\maketitle
  
\begin{abstract}
  The problem of comparing trees representing the evolutionary
  histories of cancerous tumors has turned out to be crucial, since
  there is a variety of different methods which typically infer
  multiple possible trees.  A departure from the widely studied
  setting of classical phylogenetics, where trees are leaf-labelled,
  tumoral trees are fully labelled, \ie, \emph{every} vertex has a
  label.

  In this paper we provide a rearrangement distance measure between
  two fully-labelled trees.  This notion originates from two
  operations: one which modifies the topology of the tree, the other
  which permutes the labels of the vertices, hence leaving the
  topology unaffected.  While we show that the distance between two
  trees in terms of each such operation alone can be decided in
  polynomial time, the more general notion of distance when both
  operations are allowed is NP-hard to decide.  Despite this result,
  we show that it is fixed-parameter tractable, and we give a
  4-approximation algorithm when one of the trees is binary.
\end{abstract}

\ifdefined \draftmode
\clearpage
\setcounter{page}{1}
\fi

\section{Introduction}
\label{section:introduction}

Tree rearrangement concerns modifying the topology of a set of
elements arranged in a tree and has a large
literature~\cite{semple2003phylogenetics} on phylogenies, that is,
trees whose leaves, and only leaves, are labelled by \emph{taxa}.
This setting is convenient, since it is a faithful model of the actual
problem that biologists want to solve: to find a plausible
evolutionary history that can explain a set of extant species (or
individuals)~\cite{felsenstein:inferring-phylogenies}
--- the internal nodes being the hypothetical ancestral taxa, which
are extinct.  The problem of inferring phylogenies is centuries old,
and there is a rich literature on computational methods, which fall
into major groups such as parsimony
methods~\cite{farris-1970-methods,fitch-1971-toward}, or
maximum-likelihood
methods~\cite{chor-2005-maximum,felsenstein:inferring-phylogenies}.
With the wealth of different methods for inferring phylogenetic trees
on a given set of taxa, a notion of \emph{rearrangement distance}
between output trees can be useful to assess the reliability of
methods or even the data itself in inferring such trees.  This sparked
a body of research on rearrangement distances for phylogenetic trees,
resulting in rearrangement distance measures such as \emph{nearest
  neighbor interchange} (NNI)~\cite{dasgupta-1997-on}, \emph{subtree
  prune and regraft} (SPR)~\cite{bordewich-2005-on} and \emph{tree
  bisection and reconnection} (TBR)~\cite{allen-2001-subtree} ---
interestingly, all such rearrangement distances are NP-hard to decide.
See~\cite{kuhner-2014-practical} for a comprehensive survey of the
above rearrangement distance measures, as well as other more general
distance measures for phylogenetic trees, such as the classical
Robinson-Foulds distance~\cite{rf-distance-1979}.

This paper focuses on a different, albeit related, biological field:
the study of cancer progression.  While the theory on cancer
progression as an evolutionary process is several decades
old~\cite{nowell-1976-the}, only in the past decade are data appearing
that allow us to reconstruct in detail the evolutionary history of the
progression of various cancers~\cite{Jiao2014,Hajirasouliha2014} ---
providing insight into drug resistance and devising therapeutic
strategies~\cite{Morrissy2016,Wang2016}.  In this setting, we have one
or more tumor samples where the taxa are cancer
\emph{clones}~\cite{nowell-1976-the,Hajirasouliha2014}, or groups of
cancer cells at various stages of mutation --- all of which originate
from a single \emph{driver mutation}, and the goal is again to
construct the most likely evolutionary history of these clones.  The
key difference in this setting is that --- since the tumor is only
months old --- all of the clones, even the one representing only the
driver mutation, are present in the samples, \ie, the internal nodes
are extant taxa.  In this setting, the inferred evolutionary history
is rather a \emph{fully-labelled} tree, where a label represents a
single mutation that has been acquired by a clone during evolution. It
is quite common to assume that the evolution of mutations follows the
{\em infinite sites
  assumption}~\cite{Jahn2016,Hajirasouliha2014,el-kebir_reconstruction_2015}
which implies that once a mutation is acquired in a node it is never
lost, and thus it will be present in all the clones associated with
the descendants of the node.  The above assumption motivates the fact
that we can label each node with a single mutation.  More recently,
this assumption has been challenged~\cite{Kuipers094722}.

There is already a wealth of methods for inferring such fully-labelled
cancer evolutionary trees, most of them leveraging bulk sequencing
data~\cite{Jiao2014,Hajirasouliha2014,Yuan2015,el-kebir_reconstruction_2015,bonizzoni2017beyond,DBLP:conf/iccabs/CiccolellaGPVHB18},
however methods taking advantage of higher resolution Single Cell
Sequencing (SCS) technologies~\cite{Jahn2016,Ross2016} --- even some
hybrid methods --- are beginning to appear~\cite{Salehi2017}.  With
the amount of data and methods becoming available for inferring cancer
evolution, a main challenging problem turns out to be the comparison
of the multiple trees that are produced by a single method or by
different approaches, see, \eg,~\cite{popic-lichee}.  The
investigation of operations for defining the rearrangement distance
between trees output by these methods is still in its infancy and
require the comparison of trees over the same set of mutations, \ie,
labels.  Indeed, most recent works, \eg,~\cite{govek-2018-consensus}
are mainly focused on defining a consensus tree on path, or on
ancestor-based distance measures, rather than on transforming a tree
into another considering also the topology of the trees and the
ancestor-descendant relationship of the labels --- since \emph{all}
nodes of the trees are now labelled.

While we concentrate in this work on rearrangement distance between
fully-labelled trees, this move to fully-labelled trees opens up the
discussion to the more general \emph{edit distance} between
fully-labelled trees, purportedly introduced
in~\cite{tai-1979-correction}, and of which there is a sizeable
literature, \eg,~\cite{zhang-1989-editing,mcvicar-2016-sumoted}.
There is even a comprehensive survey on the topic
in~\cite{bille-2005-survey}, while a recent implementation is reported
in~\cite{pawlik-2015-edit,pawlik-2015-edit}.  Even in the context of
cancer progression, a recent paper~\cite{DBLP:conf/wabi/KarpovMRS18}
provides a notion of edit distance for multi-labelled trees has been
defined with the goal of reconciling two trees over distinct sets of
labels into a common one.


In this work, we open the investigation of some notions of the
rearrangement distance for two rooted trees which are fully labelled
by the same set of labels.  Following the existing
literature~\cite{semple2003phylogenetics,steel_phylogeny:2016} on
phylogeny rearrangement, we extend to several operations for
rearranging a fully-labelled tree.  The distance between a pair of
trees is then the shortest sequence of these operations that
transforms the first tree into the second tree.  The first operation
we introduce is an adaptation of the SPR
operation~\cite{bordewich-2005-on} to a fully-labelled tree.  We
introduce a second operation that consists of a permutation of the
labels of the tree --- notice that such an operation does not really
make sense on leaf-labelled phylogenies.  For both operations, we
provide an algorithm for computing the shortest sequence of operations
needed to transform an input tree into a second input tree.  Then we
extend this rearrangement measure by allowing both operations: we show
that the new computational problem of finding a shortest sequence of
operations is NP-hard, but we give a $4$-approximation ratio and a
fixed parameter algorithm.


%
\section{Preliminaries}
\label{section:preliminaries}

A \emph{tree} is an undirected connected graph $T = (V_T, E_T)$
without cycles: its degree-one vertices are called \emph{leaves},
while the remaining vertices are called \emph{internal} vertices.
Trees $T_1$ and $T_2$ are is \emph{isomorphic}, and we write $T_1
\cong T_2$, if there is a bijective (or one-to-one) mapping $m:
V_{T_1} \rightarrow V_{T_2}$ such that $(u,v) \in E_{T_1}$ iff $(m(u),
m(v)) \in E_{T_2}$.  Such a mapping is referred to as an
\emph{isomorphic mapping}, or an \emph{isomorphism}.

A \emph{rooted} tree has an edge between some vertex $w \in V_T$ and
an extra \emph{root} vertex $\lambda_T \not \in V_T$ which has been
added, implicitly directing the edges, \eg, away from the root.  We
can hence define the \emph{parent} and \emph{child} relationships,
$p_T: V_T \rightarrow V_T \cup \{\lambda_T\}$ and
$c_T: V_T \cup \{\lambda_T\} \rightarrow 2^{V_T}$ respectively, where
$p_T(v) =u$ (resp., $v \in c_T(u)$) if
$(u,v) \in E_T \cup \{(\lambda_T, w)\}$ and $u$ is on the path from
$\lambda_T$ to $v$.  Note that since the children $c_T(u)$ of some
vertex $u$ is a \emph{set}, hence they are \emph{unordered}, unlike in
some notions defined in~\cite{bille-2005-survey} where an ordering can
be specified.  Moreover, $|c_T(u)|$ for any $u$ is not of any fixed
size, \ie, vertices are of \emph{unbounded degree} --- the trees are
not necessarily binary, for example.  More generally, we say that a
node $u$ is an \emph{ancestor} of a node $v$ if $u$ is on the path
from $\lambda_T$ to $v$, and conversely, $v$ is a \emph{descendant} of
$u$.  We extend the above notion of isomorphism to a pair $T_1$, $T_2$
of rooted trees by adding the condition that
$m(\lambda_{T_1}) = \lambda_{T_2}$.

A tree $T$ is \emph{fully labelled} when its vertices $V_T$ are in a
one-to-one correspondence with some set $\cL$ of labels, implying that
$|V_T| = |\cL|$.  Since all trees in this paper are rooted and fully
labelled, we will henceforth use the term \emph{tree} to denote a
rooted, fully labelled tree, and the expression \emph{tree $T$
  labelled by $\cL$} to denote a rooted tree $T$, fully labelled by
the set $\cL$ of labels.  Trees $T_1$ and $T_2$ each labelled by $\cL$
are \emph{congruent} if $T_1$ and $T_2$ are isomorphic, and one of the
isomorphisms $m$ also has the property that for every $u \in V_{T_1}$,
$u$ and $m(u)$ have the same label.  Note that if $T_1$ and $T_2$ are
congruent, then the unique child of $\lambda_{T_1}$ has the same label
as the unique child of $\lambda_{T_2}$.  Since all vertices of $T_1$
and $T_2$ are uniquely labelled and $\lambda_{T_1}$, $\lambda_{T_2}$
are special, we refer to a vertex of a tree and its label
interchangeably, when this has no effect, while clearly distinguishing
them in contexts where it matters.

In this paper, we study some notions of distance between pairs $T_1$,
$T_2$ of trees based on some \emph{rearrangement operations} that
transform $T_1$ into some $T_1'$ that is congruent to $T_2$. For the
sake of simplicity, from now on we will slightly abuse terminology in
saying that a sequence of operations transforms $T_1$ into $T_2$.
Given a tree $T$ labelled by $\cL$, the operations are:
\vspace*{.1cm}
\begin{itemize}
\item \textbf{link-and-cut operation}: given labels $v$, $p_T(v) = u$
  and a third label $w$ which is not a descendant of $v$, remove the
  edge $(u,v)$ and add the edge $(w,v)$, effectively switching the
  parent $p_T(v)$ of $v$ from $u$ to $w$.  We denote this operation
  $\edge{v}{u}{w}$.
\item \textbf{permutation operation}: apply some permutation $\pi :
  \cL \rightarrow \cL$ to the labels of $V_T$.  Each label $v \in \cL$
  of $T$ will have the new label $\pi(v)$ after this operation.
\end{itemize}
\vspace*{.1cm}
Notice that the link-and-cut operation modifies the topology of the
tree, while the permutation operation shuffles the labels without
affecting the topology.  The link-and-cut operation is so called,
following the terminology of~\cite{sleator-1981-dynamic} --- in this
article they define the two operations separately, while ours is a
certain combination of them.  Our link-and-cut operation is quite
similar to the \emph{subtree moving} (\emph{edit distance}) operation
of~\cite{mcvicar-2016-sumoted}, however the operation here has the
constraint that the new parent $w$ of child $v$ must be within a
certain distance in the tree from the original parent $u$ --- the only
restriction on our link-and-cut operation is that $w$ cannot be a
descendant of $v$.  Both of these operations are invertible: if an
operation $\sigma$ transforms $T$ into $T'$, then its inverse
operation $\sigma^{-1}$ transforms $T'$ into $T$.  For example, if
$T'$ was obtained by applying $\edge{v}{u}{w}$ to $T$, then applying
$\edge{v}{w}{u}$ to $T'$ results in $T$.  Similarly, if $T'$ was
obtained by applying $\pi$ to $T$, then applying $\pi^{-1}$ to $T'$
results in $T$.  By induction, any sequence of the above operations is
invertible.

\begin{figure}

  \tikzset{
    level 2/.style={
      sibling distance=3cm,
    },
    level 3/.style={
      sibling distance=1.2cm,
    },
  }

  \centering
  \begin{tikzpicture}
    \node (lambda) {$\lambda$}
    child { node {$a$}
      child { node {$b$}
        child { node {$d$} }
        child { node {$e$} }
        child { node {$f$} } }
      child { node {$c$}
        child { node {$g$} }
        child { node {$h$} } } } ;
    \node[draw=none] at (0,-5.5) {\Large $T_1$} ;
  \end{tikzpicture}
  \hspace*{1.2cm}
  \begin{tikzpicture}
    \node (lambda) {$\lambda$}
    child { node {$a$}
      child { node {$d$}
        child { node {$b$} }
        child { node {$e$} } }
      child { node {$c$}
        child { node {$g$} }
        child { node {$f$} }
        child { node {$h$} } } } ;
    \node[draw=none] at (0,-5.5) {\Large $T_2$} ;
  \end{tikzpicture}

  \caption{Trees $T_1$ and $T_2$ labelled by $\cL =
    \{a,b,c,d,e,f,g,h\}$.  The link-and-cut distance
    $d_{\ell}(T_1,T_2)$ between $T_1$ and $T_2$ is 4 --- the sequence
    $\edge{d}{b}{a}; ~\edge{e}{b}{d}; ~\edge{f}{b}{c};
    ~\edge{b}{a}{d}$ being an example of a smallest sequence of such
    operations transforming $T_1$ into $T_2$.  Notice that the
    operation $\edge{b}{a}{d}$ only becomes valid after performing
    $\edge{d}{b}{a}$.  The permutation distance is $d_{\pi}(T_1,T_2) =
    6$, for example, $\pi = (b~g~d~c)(e~h)$.  Finally, the
    rearrangement distance is $d(T_1,T_2) = 3$, for example, $(b~d);
    ~\edge{f}{d}{c}$.}
  \label{fig:dist}
\end{figure}

Additionally, by the definition of permutation, a sequence $\cS =
\pi_1, \dots, \pi_k$ of permutation operations can be expressed as the
single $\pi = \pi_k \cdot \pi_{k-1} \cdots \pi_2 \cdot \pi_1$ --- the
composition of the permutations of sequence $\cS$.
Finally, link-and-cut and permutation operations are interchangeable
in a sequence of operations, in the sense that the application of the
link-and-cut operation $\edge{v}{u}{w}$ followed by a permutation
$\pi$, has the same effect as $\pi$ followed by the link-and-cut
operation $\edge{\pi(v)}{\pi(u)}{\pi(w)}$ --- inspect
Figure~\ref{fig:dist} and Example~\ref{ex:interchange}.  We have the
following important property.

\begin{lemma}
  \label{lemma:permutation-first}
  Let $T_1$ and $T_2$ be each labelled by $\cL$, and $\cS = \sigma_1,
  \ldots, \sigma_k$ be a sequence of $k$ operations that transforms
  $T_1$ into $T_2$.  Then there exists a sequence $\cS^* = \sigma_1^*,
  \ldots, \sigma_k^*$ of $k$ operations in which all permutation
  operations precede all link-and-cut operations, and $\cS^*$
  transforms $T_1$ into $T_2$.  We say that $\cS^*$ is (sequence
  $\cS$) in \emph{canonical form}.
\end{lemma}

\begin{proof}
  We continue to swap consecutive pairs $\sigma_i$, $\sigma_{i+1}$, $1
  \le i < k$, of operations in sequence $\cS$, where $\sigma_i =
  \edge{u}{v}{w}$ is a link-and-cut operation and $\sigma_{i+1} = \pi$
  is a permutation operation, with the pair $\sigma_{i+1}$, $\sigma_i'
  = \edge{\pi(v)}{\pi(u)}{\pi(w)}$, until we obtain a sequence $\cS^*$
  in which all permutation operations precede all link-and-cut
  operations.  By induction on this interchange, the resulting $\cS^*$
  is of length $k$ and transforms $T_1$ into $T_2$.
\end{proof}

\begin{example}
  \label{ex:interchange}
  Consider $T_1$ and $T_2$ of Figure~\ref{fig:dist}. The application
  of the link-and-cut operation $\edge{f}{b}{c}$ followed by the
  permutation $\pi=(b~d)$, has the same effect as $\pi$ followed by
  the link-and-cut operation
  $\edge{\pi(f)}{\pi(b)}{\pi(c)}=\edge{f}{d}{c}$.
\end{example}

We now give the following notions of \emph{distance} between trees
$T_1$ and $T_2$ labelled by $\cL$.
\begin{definition}[link-and-cut distance]
  \label{definition:edge}
  The \emph{link-and-cut distance} $d_{\ell}(T_1,T_2)$ is the length
  of the shortest sequence of link-and-cut operations which transforms
  $T_1$ into $T_2$.
\end{definition}
The following Lemma ensures that the definition of link-and-cut
distance is well posed.  See also Figure~\ref{fig:dist}.
\begin{lemma}
\label{lemma:wellposed}
Given trees $T_1$ and $T_2$ each labelled by $\cL$, there always
exists a sequence of link-and-cut operations that transforms $T_1$
into $T_2$.
\end{lemma}
\begin{proof}
For any node $v$, $p_{T_1}(v)=u$, such that $p_{T_2}(v)=w$ and $w$ is
a descendant of $v$ in $T_1$ --- and thus the operation
$\edge{v}{u}{w}$ is not directly applicable --- we prove that there
exists a node $z$ on the path from $v$ to $w$ in $T_1$ (including $w$)
such that $p_{T_2}(z)$ is not a descendant of $v$ in $T_2$ nor a
descendant of $z$ in $T_1$.  This implies that after applying the
valid operation $\edge{z}{p_{T_1}(z)}{p_{T_2}(z)}$, the operation
$\edge{v}{u}{w}$ becomes valid too.  There is always such a node $z$
because, should it not exist, $w$ would be a descendant of $v$ also in
$T_2$, giving rise to the cycle $(w \rightarrow v \rightarrow \cdots
\rightarrow w)$ and thus contradicting the fact that $T_2$ is a tree.
\end{proof}
Note that, given a \emph{permutation} $\pi$ of some set $S$ of
elements, we denote its \emph{size} $|\pi|$ as the number of elements
perturbed by $\pi$, \ie, the size of the set $\{ s \in S ~:~ \pi(s)
\ne s \}$.
\begin{definition}[permutation distance]
  \label{definition:permutation}
  The \emph{permutation distance} $d_{\pi}(T_1, T_2)$ is the size
  $|\pi|$ of the smallest permutation $\pi$ that transforms $T_1$ into
  $T_2$.
\end{definition}
Finally, we also define the \emph{size} $|\cS|$ of a sequence $\cS$ of
rearrangement operations as the size of the permutation obtained by
composing the permutations of $\cS^*$ plus the length of the sequence
of link-and-cut operations of $\cS^*$, where $\cS^*$ is sequence $\cS$
in canonical form.
\begin{definition}[rearrangement distance]
  \label{definition:rearrangment}
  The \emph{rearrangement distance} $d(T_1,T_2)$ is the smallest size of any 
  sequence of operations that transforms $T_1$ into $T_2$.
\end{definition}
%
Clearly, the permutation distance $d_{\pi}(T_1,T_2)$ is defined only
when $T_1$ and $T_2$ are isomorphic, and it is evidently well
posed. As a direct consequence of this and
Lemma~\ref{lemma:wellposed}, the definition of rearrangement distance
is also well posed.  Moreover, since these operations are invertible,
all the above distance measures are symmetric, and they satisfy by
definition the triangle inequality: consider, \eg, the rearrangement
distance. Given $T_1$, $T_2$ and $T_3$ labelled by the same set of
labels, let $\cS_{12}$ be a sequence that transforms $T_1$ into $T_2$
such that $|\cS_{12}|=d(T_1,T_2)$, $\cS_{23}$ a sequence that
transforms $T_2$ into $T_3$ with $|\cS_{23}|=d(T_2,T_3)$, $\cS_{13}$ a
sequence that transforms $T_1$ into $T_3$ and $|\cS_{13}|=d(T_1,T_3)$.
It is evident that the concatenation $\cS_{12} \cS_{13}$ of $\cS_{12}$
and $\cS_{23}$ is a sequence that transforms $T_1$ into $T_3$, and by
Definition~\ref{definition:rearrangment} its size is larger or equal to
$d(T_1,T_3)$: thus $d(T_1,T_2) + d(T_2,T_3) \ge |\cS_{12} \cS_{23}|\ge
d(T_1,T_3)$.  A similar argument shows that the triangular inequality
also holds for the link-and-cut distance and the permutation distance.

We now have the important structures.

\begin{definition}[active set]
  \label{definition:active}
  Given trees $T_1$ and $T_2$ each labelled by $\cL$, we call
  \emph{active} the subset $\cX \subseteq \cL$ of labels which have
  different parents in $T_1$ and $T_2$, \ie, $v \in \cX$ iff
  $p_{T_1}(v) \ne p_{T_2}(v)$.
\end{definition}

Given trees $T_1$ and $T_2$ each labelled by $\cL$, for each vertex
$v$ of the active set $\cX$, we can associate with $v$ the pair
$(p_{T_1}(v),p_{T_2}(v))$ of the parents of $v$ in the two trees.  Let
$\cP_{(u,w)}$ be the set $\{v ~:~ p_{T_1}(v) = u, ~ p_{T_2}(v) = w\}$
--- since each vertex has exactly one parent in each tree, each vertex
$v\in \cX$ belongs to one and only one set $\cP_{(u,w)}$.  This fact is
formalized in the following definition and illustrated in
Example~\ref{ex:partition}.

\begin{definition}[family partition]
  \label{definition:partition}
  Let trees $T_1$ and $T_2$ each be labelled by $\cL$: for each vertex
  $v \in \cX$ we denote the set $\cP_{(u,w)} = \{v ~:~ p_{T_1}(v) = u,
  ~ p_{T_2}(v) = w\}$.  Then $\cP$ is the partition of set $\cX$ into
  the nonempty sets $\cP_{(u,w)}$, $u,w \in V$.  Partition $\cP$ is
  called the \emph{family partition} of the active set $\cX$, and we
  denote its \emph{size} $|\cP|$ as the number of different
  (non-empty) subsets $\cP_{(u,w)}$ it is composed of.
\end{definition}

\begin{example}
  \label{ex:partition}
  Consider $T_1$ and $T_2$ of Figure~\ref{fig:dist}. The active set is
  $\cX=\{b,d,e,f\}$.  The family partition is composed of the
  following sets: $\cP_{(a,d)}=\{b\}$, $\cP_{(b,a)}=\{d\}$,
  $\cP_{(b,d)}=\{e\}$, $\cP_{(b,c)}=\{f\}$.
\end{example}

Note that the family partition encodes the elements of any shortest
sequence of link-and-cut operations for transforming $T_1$ into $T_2$:
$v \in \cP_{(u,w)}$ corresponds to operation $\edge{v}{u}{w}$. It is
easy to see, from the proof of Lemma~\ref{lemma:wellposed}, that a
shortest sequence of valid link-and-cut operations can be obtained from
$\cP$ by ordering the set of operations it encodes with respect to a
\emph{depth-first traversal} (DFT) of $T_1$:
$\edge{u}{p_{T_1}(u)}{p_{T_2}(u)}$ precedes
$\edge{v}{p_{T_1}(v)}{p_{T_2}(v)}$ if $u$ precedes $v$ in a DFT of
$T_1$. Hence $d_{\ell}(T_1,T_2)=|\cX|$, \ie, the link-and-cut distance
is equal to the cardinality of the active set, of which $\cP$ is a
partition.

\section{Computational complexity}
\label{section:complexity}

In this section we determine the complexity of computing the distance
between two trees labelled by the same set of labels, in terms of the
three distance measures defined in
Section~\ref{section:preliminaries}.  More precisely, despite the fact
that the link-and-cut and the permutation distances are
polynomial-time computable, computing the rearrangement distance is
NP-hard.

\subsection{Link-and-cut distance}
\label{subsection:edge}

We first show that we can compute the link-and-cut distance between
two trees in linear time by showing that the family partition can be
built in linear time.

\begin{lemma}
  \label{lemma:edge}
  The link-and-cut distance $d_{\ell}(T_1,T_2)$ between trees $T_1$
  and $T_2$ each labelled by $\cL$ can be computed in time $O(|\cL|)$.
\end{lemma}

\begin{proof}
  Since the link-and-cut distance is $|\cX|$, it suffices to
  demonstrate that the family partition $\cP$ can be built in time
  $O(|\cL|)$.  The procedure is as follows: we first do a DFT of tree
  $T_1$, building an array $p_{T_1}(v)$ of the parents in $T_1$,
  indexed by the child $v$.  We build the same array $p_{T_2}(v)$ for
  tree $T_2$.  Then we go through the set $\cL$ of labels, in some
  order: at each label $v$, should $p_{T_1}(v) = u \ne p_{T_2}(v) =
  w$, we add $v$ to $\cP_{(u,w)}$ of the family partition $\cP$.  Then
  we just sum up the sizes of the non-empty subsets $\cP_{(u,w)}$ of
  $\cP$ in order to obtain $|\cX| = d_{\ell}(T_1,T_2)$.  Clearly each
  tree traversal can be done in time $O(|\cL|)$, because $|V_{T_1}| =
  |V_{T_2}| = |\cL|$.  In going through the labels, for each label
  $v$, we either add or do not add the single vertex $v$ to $\cP$, and
  so this procedure takes time $O(|\cL|)$.
\end{proof}

\subsection{Permutation distance}
\label{subsection:permutation}

This subsection is dedicated to proving the following lemma, which
shows that we can compute the permutation distance between two trees
in cubic time.

\begin{lemma}
  \label{lemma:permutation}
  The permutation distance $d_{\pi}(T_1, T_2)$ between isomorphic
  trees $T_1$ and $T_2$ each labelled by $\cL$ can be computed in time
  $O(|\cL|^3)$.
\end{lemma}

We need the following definitions and auxiliary lemmas.  The
\emph{mismatch number} $\Delta(m)$ of some isomorphic mapping $m$ from
tree $T_1$ to tree $T_2$ is the number of vertices whose label is not
conserved by $m$.  More formally, $\Delta(m) = \left| \left\{ u \in
V_{T_1} ~:~ \ell_{T_1}(u) \ne \ell_{T_2}(m(u)) \right\} \right|$,
where $\ell_T : V_T \rightarrow \cL$ is the the one-to-one
correspondence between the vertices $V_T$ of tree $T$ and the set
$\cL$ of labels.  Let $\cI(T_1,T_2)$ be the set
of isomorphic mappings from $T_1$ to $T_2$.  Given isomorphic trees
$T_1$ and $T_2$ each labelled by $\cL$, let the \emph{mismatch
  distance} $\dm(T_1, T_2) = \min \{ \Delta(m) ~:~ m \in \cI(T_1,T_2)
\}$ be the minimum mismatch number of an isomorphic mapping from $T_1$
to $T_2$.
The following equality between permutation distance and mismatch
distance holds.

\begin{lemma}
  \label{lemma:permutation=mismatch}
  Given isomorphic trees $T_1$ and $T_2$ each labelled by $\cL$, it
  follows that $d_{\pi}(T_1, T_2) = \dm(T_1, T_2)$.
\end{lemma}

\begin{proof}
  Consider an isomorphic mapping $m$ from $T_1$ to $T_2$ that has the
  minimum mismatch number $\Delta(m) = \dm(T_1,T_2)$ and let $\cL'
  \subseteq \cL$ be the set of labels of the vertices involved in the
  set of mismatching vertices between $V_{T_1}$ and $V_{T_2}$ given by
  $m$.

  Clearly, such labels are in a permutation $\pi$ which rearranges the
  labels of tree $T_1$ to obtain $T_2$, while, by construction of $m$,
  all the other labels will not be perturbed by $\pi$.  Then we need
  to show that such a permutation rearranges the minimum number of
  distinct labels, that is, its size $|\pi| = d_{\pi}(T_1, T_2)$ is
  the permutation distance.  Indeed, assume to the contrary that the
  permutation distance $d_{\pi}(T_1,T_2) < |\pi|$.  This implies the
  existence of a permutation $\pi'$ that rearranges fewer labels than
  $\pi$, \ie, $|\pi'| < |\pi|$.  Then we show that there exists an
  isomorphic mapping $m'$ that has mismatch number less than the one
  of $m$, contradicting the initial assumption.

  Indeed, consider the permutation $\pi'$ and define the mapping $m'$
  from $T_1$ to $T_2$ such that $m'(u) = v$ whenever $\pi'(v) = u$ The
  mapping $m'$ is an isomorphism by the construction of $\pi'$, since
  $m'(\pi'(u)) = u$ for all $u \in V_{T_1}$, and with the application
  of $\pi'$, the two trees are congruent, hence congruency of the
  labels implies isomorphism of the two trees.  This concludes the
  proof that $m'$ is an isomorphism thus leading to a contraction.
\end{proof}

Now the task is to show how we can efficiently find an isomorphic
mapping $m \in \cI(T_1,T_2)$ from tree $T_1$ to tree $T_2$ such that
the mismatch number $\Delta(m)$ is minimized --- in other words, we
need to compute $\dm(T_1, T_2)$.  For tree $T$, let $L_T$ be its set
of leaves, and for vertex $u \in V_T$, let $T|u$ be the subtree of $T$
\emph{rooted} at $u$ --- the connected component of $T'$ containing
$u$, where $T'$ is the tree obtained from $T$ by removing the edge
$(p_T(u),u)$, and $p_T(u)$ is the \emph{parent} of $u$ in $T$.  If $u$
and $v$ are both vertices, we slightly abuse notation using
$\Delta(u,v)$ to mean the mismatch distance of (the only possible
mapping between) $u$ and $v$, that is $\Delta(u,v) = 0$ if
$\ell_{T_1}(u) = \ell_{T_2}(v)$ --- that is, the vertex $u$ of $T_{1}$
and the vertex $v$ of $T_{2}$ have the same label --- and $\Delta(u,v)
= 1$ otherwise.  Recall that the \emph{children} of $u$ in $T$ is the
set $c_T(u) = \{ v ~:~ (u,v) \in E_{T|u} \}$ --- let $\cC_{u,v}$ be
the set of bijective mappings $m : c_{T_1}(u) \rightarrow c_{T_2}(v)$
from the children of $u \in V_{T_1}$, $v \in V_{T_2}$ --- clearly any
isomorphic mapping in $\cI(T_1,T_2)$ is such a mapping when restricted
to $c_{T_1}(u)$ and $c_{T_2}(v)$.  We define the following (recursive)
relationship $D(u,v)$, showing later how we can use it to compute the
mismatch distance:

\[
  D(u,v) = \left\{%
    \begin{array}{ll}
      \Delta(u,v) & \textrm{ if } u \in L_{T_1}, ~ v \in L_{T_2},\\

      \min\limits_{m \in \cC_{u,v}} \sum_{z \in c_{T_1}(u)} D(x,m(z))
      + \Delta(u,v) & \textrm{ if } T_1|u \cong T_2|v, u \not \in
      L_{T_1}, v \not \in L_{T_2},\\

      \infty & \textrm{ otherwise.}\\
    \end{array}
  \right.
\]

We first show how to compute $T_1|u \cong T_2|v$ (true or false) for
all $u \in V_{T_1}$, $v \in V_{T_2}$, since we need it for computing
$D(u,v)$.  We can build this relationship by extending the time $O(n
\log n)$ algorithm of~\cite{campbell-1991-isomorphism}, where
$n=|V_{T_1}|=|V_{T_2}|$, for determining if two rooted trees are
isomorphic.\footnote{Note that there is a linear time algorithm
  in~\cite{ahu-1974}, but it assumes that $\log n$ is fixed, where $n$
  is the number of vertices --- which is the likely the case for all
  practical instances.}  The idea of this algorithm is to first
organize each tree in \emph{levels}, where the level of a vertex is
its distance from the root --- this can be done with a simple DFT of
each tree.  Suppose each tree has $k$ levels, otherwise they do not
have the same number of levels, hence they cannot be isomorphic.
Starting from level $k$ in both trees, we move up level by level
towards the root in both trees simultaneously.  At each level $i$ we
perform the following steps: (1) for each vertex $u$ of each tree $T$
on level $i$, we store a representation of the topology of $T|u$,
computing it recursively from the representations stored in the
children of $u$ on level $i+1$; then (2) sort in each tree the
vertices at level $i$ by representation\footnote{We assume that there
  is a total ordering on the representations of the topologies, they
  are of constant size, and can be compared in constant time --- for
  details see \cite{campbell-1991-isomorphism,ahu-1974}}; and finally
(3) compare the two resulting sorted orders --- only when they are
identical, may we proceed to the next level.  If we make it all the
way to the first level (the root), and we succeed with the 3 steps at
this level, then the two trees are isomorphic, otherwise not.  We can
extend this algorithm with a fourth step: (4) for each pair $u \in
V_{T_1}$, $v \in V_{T_2}$ of vertices on level $i$, if their
representations are identical, then $T_1|u \cong T_2|v$ is true, and
false otherwise.  Clearly, only for pairs of vertices on the same
level, can their subtrees be isomorphic, and so this is an exhaustive
search for all such pairs.  Since the number of vertices compared in
step (4) over all of the levels is no more than $n^2$, this
computation of $T_1|u \cong T_2|v$ for all pairs $u \in V_{T_1}$, $v
\in V_{T_2}$ of vertices requires time $O(n^2)$.  We are now ready to
prove that $D(u,v) = \dm(T_1|u, T_2|v)$ for each vertex $u$ of $T_{1}$
and $v$ of $T_{2}$.

\begin{lemma}
  \label{lemma:mismatch}
  Let trees $T_1$ and $T_2$ each be labelled by $\cL$.  Then (1)
  $D(u,v) = \dm(T_1|u,T_2|v)$ for all pairs $u \in V_{T_1}$, $v \in
  V_{T_2}$.  Moreover, (2) $D(u,v)$ can be computed in time
  $O(|\cL|^3)$.
\end{lemma}

\begin{proof}
  (1) It is essentially a proof by induction.  If both $u$ and $v$ are
  leaves, then $T_1|u \cong T_2|v$ is trivially true and
  $D(u,v) = \Delta(u,v)$.  When $T_1|u \cong T_2|v$ does not hold,
  then the mismatch distance is undefined.

  Otherwise, if $u$ and $v$ are internal vertices, and
  $T_1|u \cong T_2|v$, then let $m$ be a bijective mapping from the
  nodes of $T_1|u$ to the nodes of $T_2|v$ minimizing the mismatch
  distance.  By the definition of $\dm$ and the construction of $m$,
  $\dm(T_1|u,T_2|v) = \sum_{z \in c_{T_1}(u)} \dm(T_1|z,T_2|m(z)) +
  \Delta(u,v)$.  By the inductive hypothesis
  $\sum_{z \in c_{T_1}(u)} \dm(T_1|z,T_2|m(z)) = \sum_{z \in
    c_{T_1}(u)} D(z,m(z))$.  Combining these two facts implies that
  $\dm(T_1|u,T_2|v) = \sum_{z \in c_{T_1}(u)} D(z,m(z)) +
  \Delta(u,v)$.  Since $m$ is the bijective mapping minimizing the
  mismatch distance, no other mapping $m'$ can achieve a smaller value
  of $\sum_{z \in c_{T_1}(u)} D(z,m(z)) + \Delta(u,v)$, hence
  $D(u,v) = \dm(T_1|u,T_2|v)$.

  (2) Computing $D(u,v)$ when $u$ and $v$ are leaves requires constant
  time.  When $u$ and $v$ are internal nodes --- assuming we have
  already computed $T_1|u \cong T_2|v$ --- we compute a \emph{minimum
    weight} matching in the weighted bipartite graph with $V =
  c_{T_1}(u) \cup c_{T_2}(v)$ and $E = \{(x,y) ~:~ T_1|x \cong T_2|y,
  ~x \in c_{T_1}(u), ~y \in c_{T_2}(v)\}$ with weight function $w : E
  \rightarrow \mathbb{N}$ such that $w(x,y) = D(x,y)$.  Such a
  matching can be found in time $O(|V| \log |V| + |V||E|)$ using a
  Fibonacci heap~\cite{Fredman:1987:FHU:28869.28874}.  If we sum over
  all of graphs in which we compute these matchings during the
  recursive computation of $D(u,v)$ for all $u \in V_{T_1}$, $v \in
  V_{T_2}$, the number of vertices and edges in each graph, these sums
  will be at most $n$ (vertices) and $n^2$ (edges) respectively, since
  both trees have $n=|V_{T_1}|=|V_{T_2}|=|\cL|$ vertices.  This means
  that this matching procedure will take overall time $O(|\cL|^3)$.
  Computing $T_1|u \cong T_2|v$ for each pair $u \in V_{T_1}$, $v \in
  V_{T_2}$ takes time $O(|\cL|^2)$ overall, hence computing computing
  $D(u,v)$ for all $u \in V_{T_1}$, $v \in V_{T_2}$ has a total
  running time $O(|\cL|^3)$.
\end{proof}

Lemma~\ref{lemma:permutation} then follows from
Lemmas~\ref{lemma:permutation=mismatch} and \ref{lemma:mismatch}.
Notice that, when computing $D(u,v)$ we can also maintain the set
$C(u,v)$ of the labels that are \emph{conserved} in the minimum weight
matchings, that is, those that are not involved in a mismatch.  More
precisely, $C(u,v)$ is equal to the union of the $C(x,y)$ over all
edges $(x,y)$ of the optimal matching.  To that set, we add the label
$\ell_{T_{1}}(u)$ if $\Delta(u,v) = 0$.  Once we have all sets
$C(u,v)$, the permutation that we want to compute involves exactly the
labels not in $C(\lambda_{T_{1}}, \lambda_{T_{2}})$.  More precisely,
the label $\ell_{T_1}(u)$ must be replaced with the label of the
vertex $m(u)$, where the isomorphism $m$ can be constructed from the
perfect matchings of the optimal solution.

\subsection{Rearrangement distance}
\label{subsection:rearrangement}

Finally, we show that deciding the rearrangement distance between two
trees is NP-hard.  We show this by reduction from 3-dimensional
matching, one of Karp's 21 NP-complete
problems~\cite{karp-1972-reducibility}.

In 3-dimensional matching, we are given three disjoint sets $A$, $B$
and $C$, along with a set $\cT$ of triples $(a,b,c)$, such that
$a \in A$, $b \in B$ and $c \in C$ --- essentially a 3-uniform
hypergraph $H$.  A \emph{matching} is then a subset
$\cM \subseteq \cT$ such that for every two triples $(a,b,c) \in \cM$,
$(a',b',c') \in \cM$, it follows that $a \neq a'$, $b \neq b'$ and
$c \neq c'$, that is, all triples of $\cM$ are pairwise disjoint.  It
is then NP-hard to decide for a given $k$ if there is a matching $\cM$
of size $k$~\cite{karp-1972-reducibility}.  It has been proved that
the problem remains NP-hard even in the case of 3-bounded 1-common
3-dimensional matching, which is a restriction of the problem where
the number of occurrences of an element in the triples is at most 3,
and each pair of triples has at most one element in
common~\cite{jiang-2015-approximation}.  We also make use of the
following structure in this proof.

\begin{definition}[movements graph]
  \label{definition:movements}
  Given trees $T_1$ and $T_2$ each labelled by $\cL$, the
  \emph{movements graph} $G$ has an edge for every element
  $\cP_{(u,w)}$ of the family partition $\cP$ of $T_1$ and $T_2$, that
  is, $E_G = \{ (u,w) ~:~ \cP_{(u,w)} \in \cP \}$, while its vertex
  set is $V_G = \bigcup_{(u,w) \in E_G} \{u,w\}$.
\end{definition}

We now prove that computing the rearrangement distance is NP-hard.

\begin{theorem}
  \label{theorem:rearrangement}
  Given trees $T_1$ and $T_2$ each labelled by $\cL$, and some integer
  $k$, it is NP-hard to decide if $d(T_1, T_2) \leq k$.
\end{theorem}

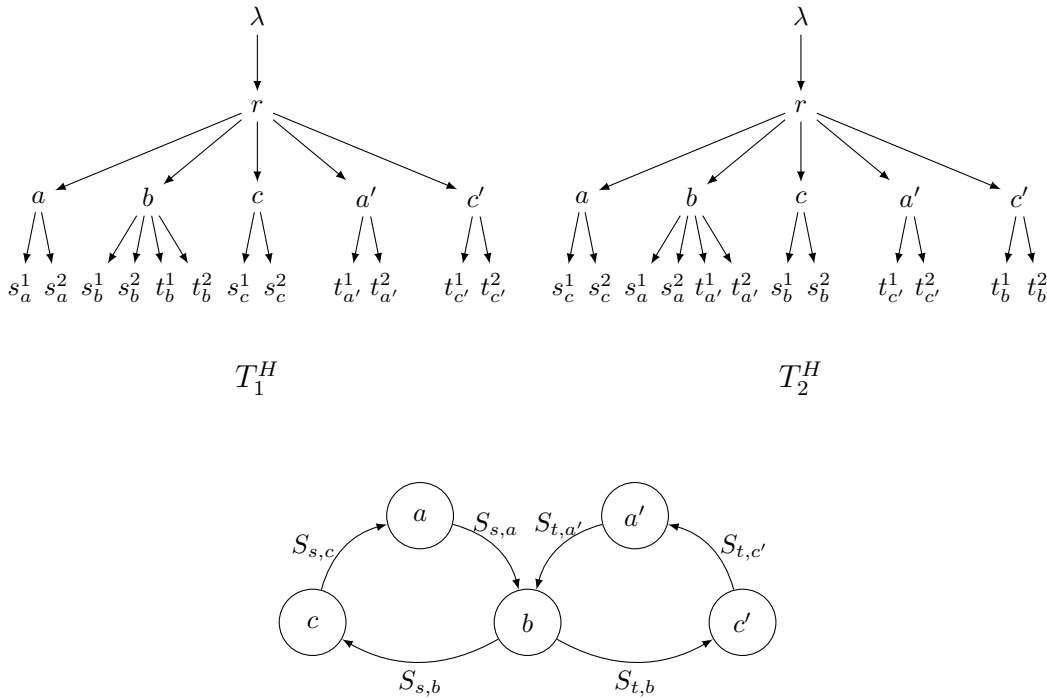
\begin{figure}[ht!]

  \tikzset{
    level 2/.style={
      sibling distance=1.8cm,
    },
    level 3/.style={
      sibling distance=.6cm,
    },
  }

  \centering
  \begin{tikzpicture}[scale=0.8]
    \node (lambda) {$\lambda$}
    child { node {$r$}
      child { node {$a$}
        child { node {$s_a^1$} }
        child { node {$s_a^2$} } }
      child { node {$b$}
        child { node {$s_b^1$} }
        child { node {$s_b^2$} }
        child { node {$t_b^1$} }
        child { node {$t_b^2$} } }
      child { node {$c$}
        child { node {$s_c^1$} }
        child { node {$s_c^2$} } }
      child { node {$a'$}
        child { node {$t_{a'}^1$} }
        child { node {$t_{a'}^2$} } }
      child { node {$c'$}
        child { node {$t_{c'}^1$} }
        child { node {$t_{c'}^2$} } } } ;
    \node[draw=none] at (0,-6) {\Large $\tone$} ;
  \end{tikzpicture}
  \hspace*{.1cm}
  \begin{tikzpicture}[scale=0.8]
    \node (lambda) {$\lambda$}
    child { node {$r$}
      child { node {$a$}
        child { node {$s_c^1$} }
        child { node {$s_c^2$} } }
      child { node {$b$}
        child { node {$s_a^1$} }
        child { node {$s_a^2$} }
        child { node {$t_{a'}^1$} }
        child { node {$t_{a'}^2$} } }
      child { node {$c$}
        child { node {$s_b^1$} }
        child { node {$s_b^2$} } }
      child { node {$a'$}
        child { node {$t_{c'}^1$} }
        child { node {$t_{c'}^2$} } }
      child { node {$c'$}
        child { node {$t_b^1$} }
        child { node {$t_b^2$} } } } ;
    \node[draw=none] at (0,-6) {\Large $\ttwo$} ;
  \end{tikzpicture}

  \vspace*{1cm}

  \begin{tikzpicture}[node distance=2cm]
    \node[state] (a)  {$a$} ;
    \node[state] (b)  [below right of=a] {$b$} ;
    \node[state] (c)  [below left of=a] {$c$} ;
    \node[state] (ap) [above right of=b] {$a'$} ;
    \node[state] (cp) [below right of=ap] {$c'$} ;

    \path
    (a) edge [bend left] node [draw=none,above] {$S_{s,a}$} (b)
    (b) edge [bend left] node [draw=none,below] {$S_{s,b}$} (c)
    (c) edge [bend left] node [draw=none,left]  {$S_{s,c}$} (a) ;

    \path
    (ap) edge [bend right] node [draw=none,above] {$S_{t,a'}$} (b)
    (b)  edge [bend right] node [draw=none,below] {$S_{t,b}$}  (cp)
    (cp) edge [bend right] node [draw=none,right] {$S_{t,c'}$} (ap) ;
  \end{tikzpicture}

  \caption{The trees $\tone$ and $\ttwo$ given instance $H$ of
    3-dimensional matching with $A=\{a,a'\}$, $B=\{b\}$ and
    $C=\{c,c'\}$ and $\cT=\{s = \{a,b,c\}, ~ t = \{a',b,c'\}\}$ (top),
    and the corresponding movements graph for the trees $\tone$ and
    $\ttwo$ (bottom).}
  \label{fig:reduction}
\end{figure}

\begin{proof}
  Reduction from 3-bounded 1-common 3-dimensional matching.  We are
  given an instance $H$ of 3-dimensional matching consisting of a set
  $\cT$ of $m$ triples $(a,b,c)$ over the disjoint sets $A$, $B$, $C$.
  We construct two trees $\tone$ and $\ttwo$ each with $|A| + |B| +
  |C| + 6m + 2$ vertices, for which the rearrangement distance
  $d(\tone, \ttwo) \leq 3n + 6(m-n)$ if and only if $H$ has a
  3-dimensional matching of size $n$.

  Consider such an instance $H$ of 3-dimensional matching as above.
  In the construction, the trees $\tone$ and $\ttwo$ each have a root
  vertex $r$, and a vertex for every element of $A$, $B$ and $C$ ---
  each of which have $r$ as the parent.  To each $v \in \{a,b,c\}$ in
  $\tone$ and triple $t$, we add a set $S_{t,v} = \{t_v^1, t_v^2\}$ of
  two (uniquely labelled) children.  In $\ttwo$, we add the sets
  $S_{t,v}$ of two children to each of these three vertices, but
  cyclically shifted, with respect to $\tone$, \ie, we add $S_{t,a}$
  to $b$, $S_{t,b}$ to $c$, $S_{t,c}$ back to $a$ again.  Note that
  this induces in the movements graph $G$ a cycle
  $\cC_t = \{(a,b), (b,c), (c,a)\}$ --- see
  Figure~\ref{fig:reduction}.  Now, observe that the movements graph
  $G$ will have cycles of length $3$ corresponding to each triple and
  two cycles may share one common vertex $v$ if the triples share
  element $v$.  By Lemma~\ref{lemma:permutation-first}, a sequence of
  operations of total size $d(\tone,\ttwo)$ will consist of a
  permutation followed by a sequence of link-and-cut operations.
  Observe that any permutation involves for each cycle $\cC_t$ an
  edge, two edges or all three edges.  Now, the rearrangement distance
  aims to solve cycles in the movements graph in the sense that after
  the operations, the movements graph has no edges.  Observe that
  given a cycle $\cC_t$ of the movements graph $G$, then the minimum
  cost rearrangement to solve $\cC_t$ consists of applying a
  permutation of size $3$ involving the three vertices of the cycle,
  thus of total cost $3$.  Observe that two cycles sharing a common
  vertex cannot both be solved by a permutation that is a cyclic shift
  of the vertices of the cycle, that is they cannot be both solved
  with cost $3$.  Moreover, permutations of vertices cannot solve more
  than one vertex of a cycle $\cC_t$ if it is not a cyclic shift of
  the vertices of $\cC_t$, as cycles do not share edges.  In case of a
  cyclic shift of two vertices of $\cC_t$, (1) a permutation of size
  $2$ and then four link-and-cut operations are required. If instead
  (2) at most a single vertex of $\cC_t$ is involved in a permutation,
  then six link-and-cut operations are required.  We now detail how
  this implies that $d(\tone, \ttwo) \leq 3n + 6(m- n)$ if and only if
  $H$ has a 3-dimensional matching of size $n$.

  ($\Rightarrow$) Assume that $d(\tone, \ttwo) \leq 3n + 6(m - n)$.
  By the above observation on how cycles of the movements graph $G$
  are solved by the sequence of permutations, cases (1) and (2) for
  solving cycles have the same cost equal to $6$.  Thus the only
  possible way to have a rearrangement distance less than or equal to
  $3n + 6 (m- n) $ is by taking $n$ disjoint cycles solved by the
  permutation operation of cost $3$. This implies a 3-dimensional
  matching of size $n$.

  ($\Leftarrow$) Now, suppose that $H$ has a 3-dimensional matching
  $\cM \subseteq \cT$ of size $n$.  This implies that there are $n$
  triples that are disjoint and thus the movements graph $G$ has $n$
  disjoint cycles. By solving each cycle with a permutation of size
  $3$, $m- n$ cycles are left in the movements graph. The remaining
  cycles, in the worst case, share common vertices with the cycles
  solved by the permutations, and thus they can be solved with a cost
  that is $6$ in the worst case.  Thus we obtain that $d(\tone, \ttwo)
  \leq 3n + 6(m - n)$, completing the proof.
\end{proof}

\section{Bounds and approximation}
\label{section:bounds}

We first give the following important lemma which states that when we
apply a permutation to the labels of $T_1$ obtaining $T'_1$, the size
of the resulting family partition $\cP'$ cannot increase or decrease
too much with respect to the size of $\cP$.

\begin{lemma}
  \label{lemma:partition}
  Given trees $T_1$ and $T_2$ with corresponding active set $\cX$ and
  family partition $\cP$, if $T_1'$ is the tree (isomorphic to $T_1$)
  resulting from the application of permutation $\pi$ of the labels of
  $T_1$, and $\cX'$ and $\cP'$ are the active set and the family
  partition of $T_1'$ and $T_2$, respectively, then $|\cP| - 2|\pi|
  \le |\cP'| \le |\cP| + 2|\pi|$.
\end{lemma}

\begin{proof}
  Let $a$ be some label of $T_1$ which has been perturbed by
  permutation $\pi$, \ie, $\pi(a) = b \ne a$.  The new family
  partition $\cP'$ is obtained from $\cP$ by means of deletions,
  insertions and substitutions of subsets.  The crucial observation is
  that such an operation will only affect the \emph{neighborhood} of
  $a$, namely the (possibly empty) set of its children $c_{T_1}(a)$
  and its parent $p_{T_1}(a)$ in $T_1$.  Let us consider each child
  $v \in c_{T_1}(a)$ first.  We have the following cases.
  \begin{itemize}
  \item($v\in\cP_{(a,b)}\subseteq\cX$): since $\pi$ makes $b$ the
    parent of $v$, which is exactly the parent of $v$ in $T_1$,
    $v \notin \cX'$ and $\cP_{(a,b)}$ will be missing from $\cP'$;
  \item ($v\in\cP_{(a,c)}\subseteq\cX$, $c\neq b$): after applying
    $\pi$, $v$ will belong to set $\cP'_{(b,c)}$, thus $\cP_{(a,c)}$
    will be replaced by $\cP'_{(b,c)}$ in $\cP'$;
  \item ($v\notin\cX$): then $v \in \cP'_{(b,a)}$ in $\cP'$, thus
    $\cP'$ might have an extra element with respect to $\cP$.
  \end{itemize}

  Consider now the possible effects of $\pi$ on $p_{T_1}(a)$. There
  are two possible scenarios:
  \begin{itemize}
  \item ($b\in \cP_{(p_{T_1}(b),p_{T_2}(b))} \subseteq \cX$): if
    $p_{T_2}(b)=p_{T_1}(a)$, then $b \notin \cX'$ and if $b$ was the
    only element of $\cP_{(p_{T_1}(b),p_{T_2}(b))}$, the latter will
    be missing from $\cP'$; else,
    $b \in \cP'_{(p_{T_1}(a),p_{T_2}(b))}$, thus
    $\cP_{(p_1(b),p_2(b))}$ will be replaced by
    $\cP'_{(p_{T_1}(a),p_{T_2}(b))}$ in $\cP'$;
  \item ($b\notin\cX$): then $b \in \cP'_{(p_{T_1}(a),p_{T_2}(b))}$ in
    $\cP'$, thus $\cP'$ might have an extra element with respect to
    $\cP$.
  \end{itemize}
  In summary, $\cP'$ is obtained from $\cP$ with up to two deletions
  and two additions of sets in the family partition for each label
  involved in the permutation $\pi$, thus the result follows.
\end{proof}

In the special case where one of the trees, \eg, $T_1$, is
\emph{binary}, \ie, each node has up to two children, we have the
following lemma connecting link-and-cut and rearrangement distance.

\begin{lemma}
  \label{lemma:apx}
  Given $T_1$ a binary tree, $T_2$ any tree, we have that
  $d_{\ell}(T_1,T_2)\le 4\cdot d(T_1,T_2)$.
\end{lemma}

\begin{proof}
  Suppose that $T_2$ is optimally obtained from $T_1$ by applying a
  permutation $\pi$ of the labels followed by a number of link-and-cut
  operations --- something we can assume in virtue of
  Lemma~\ref{lemma:permutation-first}.  Let $T'_1$ be the tree
  resulting from the application of permutation $\pi$ of the labels of
  $T_1$, $\cX'$ and $\cP'$ the active set and family partition of
  $T'_1$ and $T_2$, respectively.  By the construction of the family
  partition, the optimal number of link-and-cut operations to obtain
  $T_2$ from $T'_1$ is at least $|\cP'|$; we thus have that
  $d(T_1,T_2)= |\pi| + |\cX'| \ge |\pi| + |\cP'| $. Moreover, Lemma
  \ref{lemma:partition} says that $|\pi|\ge \frac{|\cP|-|\cP'|}{2}$,
  thus $d(T_1,T_2)\ge \frac{|\cP|-|\cP'|}{2} + |\cP'| =
  \frac{|\cP|}{2} + \frac{|\cP'|}{2} \ge \frac{|\cP|}{2}$.  Now, since
  $T_1$ is binary, each set in the family partition $\cP$ consists of
  up to two elements (the elements of $\cP_{(x,y)}$ are the ones among
  the children of $x$ in $T_1$ that becomes the children of $y$ in
  $T_2$, thus they cannot be more than the number of children of
  $x$). It follows that $|\cX|=d_{\ell}(T_1,T_2)\le 2\cdot|\cP|$,
  hence $d(T_1,T_2)\ge \frac{d_{\ell}(T_1,T_2)}{4}$.
\end{proof}

Importantly, we note that Lemma \ref{lemma:apx} states that the
link-and-cut distance algorithm provides a linear 4-approximation for
the rearrangement distance when at least one of the trees involved is
binary.  We have the following corollary from Lemma~\ref{lemma:apx}
and Lemma~\ref{lemma:edge}.



\begin{corollary}
  \label{corollary:apx}
  There exists a polynomial time $4$-approximation algorithm for the
  rearrangement distance problem for binary trees.
\end{corollary}

\section{Fixed parameter tractability}
\label{section:fpt}

This section is devoted to showing that computing the rearrangement
distance between trees $T_{1}$ and $T_{2}$ is fixed-parameter
tractable, essentially via the bounded search tree
technique~\cite{ParameterizedComplexity}.  In this case, the instance
also contains a parameter $k$: in time $O((4k)^{2k^{2}}n)$ we (1)
determine if $d(T_{1}, T_2)\le k $ and, if this is the case, (2) find
the minimum sequence of operations transforming $T_{1}$ into $T_{2}$.

The main idea of our algorithm is that, in virtue of
Lemma~\ref{lemma:permutation-first}, we can reorder the sequence of
operations that transforms $T_{1}$ into $T_{2}$ so that all
permutations precede the link-and-cut operations.  Let $T^{*}$ be the
tree obtained from $T_{1}$ using only permutations and such that we
can optimally obtain $T_{2}$ from $T^{*}$ using only link-and-cut
operations.  Then $d(T_1,T_2) = d_{\pi}(T_{1}, T^{*}) +
d_{\ell}(T^{*}, T_{2})$.  Our algorithm consists of showing that
$d_{\pi}(T_{1}, T^{*})$ is related to the size of the family
partition, and that we can compute $d_{\ell}(T^{*}, T_{2})$ in linear
time.

The main consequence of Lemma~\ref{lemma:permutation-first} here is
that we can restrict our attention to permutations first (to obtain a
tree $T^{*}$), and to link-and-cut operations afterwards.  Finding
such a tree $T^{*}$ is easier when we want to determine if the
rearrangement distance $d(T_1,T_2)$ is at most $k$.

In fact, a consequence of Lemma~\ref{lemma:partition} is that
$d(T_{1}, T_2) \geq d_{\pi}(T_{1}, T^{*})\geq |\cP|/2 $, where $\cP$
is the family partition associated with $T_{1}$ and $T_{2}$.  Notice
that any sequence of operations that transforms $T_{1}$ into $T_2$
also transforms $\cX$ into the empty set --- thus $\cP$ into the empty
partition.

Since $d(T_{1}, T_2) \geq |\cP|/2 $, the first step of our algorithm
is to compute the family partition $\cP$ of $T_{1}$ and $T_{2}$ and
verify that $k\ge |\cP|/2$.  If that inequality is not satisfied,
then, since as stated above $d_\pi(T_{1}, T^{*}) \geq |\cP|/2 $, it
would follow that $d(T_{1}, T_2) > k $.  Hence we can focus on the
instances where $k\ge |\cP|/2$, that is $|\cP|\le 2k$.  Since the
family partition is sufficiently small, we can compute all sequences
of permutations of at most $k$ labels of $\cX$ in time
$O((4k)^{2k^{2}})$.  In fact, each of the permutations involves one of
the $2^{2k}$ subsets of vertices of $\cX$, and there can be at most
$(2k)!$ permutations of a set of $2k$ elements.  Overall there are at
most $\left(2^{2k}(2k)!\right)^{k}$ such sequences: it is trivial to
organize them in a search tree that can be generated and traversed in
linear time, and some crude upper bound results in the desired time
bound.  Let $\mathcal T$ be the set of trees that are obtained by
applying to $T_{1}$ the sequence of operations corresponding to a node
of the search tree.

The second part of our algorithm is to compute $d_{\ell}(T, T_{2})$ for
each tree in $T \in \cT$, which, by Lemma~\ref{lemma:edge}, requires
$O(n)$ time for each tree, keeping track of the tree $T^{*}$
minimizing $d_{\pi}(T_{1}, T^{*}) + d_{\ell}(T^{*}, T_{2})$.  The algorithm
has therefore $O((4k)^{2k^{2}}n)$ time complexity.

\section{Open problems}
\label{section:open}

In this paper we provide a NP-hardness proof of the rearrangement
distance for trees with vertices of unbounded degree. The
computational complexity of the rearrangement distance in the case of
bounded degree trees remains an open problem.  Mainly it would be of
theoretical interest to see if it is still NP-hard for binary trees.
On the other hand, we provide a constant approximation algorithm for
the rearrangement distance of binary trees. Extending this result to
general trees is still open. Such a result could be of interest in
developing practical algorithms for comparing tumor phylogenies.

\ifdefined \draftmode
\clearpage
\fi
\subparagraph*{Acknowledgments.}

The authors wish to thank Mauricio Soto Gomez for the inspiring
discussions.





\bibliography{lbtr_ref}

\end{document}